\documentclass[fleqn,reqno]{article}
\usepackage[margin=1.37in]{geometry}
\usepackage{amsmath,amsthm,amssymb,mathrsfs}
\usepackage{enumerate}

\theoremstyle{plain}
\newtheorem{theorem}{Theorem}
\newtheorem{proposition}[theorem]{Proposition}
\newtheorem{lemma}[theorem]{Lemma}
\newtheorem{corollary}[theorem]{Corollary}

\theoremstyle{definition}

\newcommand{\R}{\mathbb{R}}
\newcommand{\C}{\mathbb{C}}

\newcommand{\cE}{\mathcal{E}}

\newcommand{\cL}{\mathcal{L}}         

\newcommand{\Lat}{\mathcal{L}}
\newcommand{\nc}{\newcommand}

\nc{\G}{\Gamma}
\nc{\g}{\gamma}
\nc{\al}{\alpha}
\nc{\be}{\beta}
\nc{\del}{\delta}
\nc{\io}{\iota}
\nc{\ka}{\kappa}
\nc{\lam}{\lambda}
\nc{\Lam}{\Lambda}
\nc{\w}{\omega}
\nc{\Om}{\Omega}
\nc{\Oms}{\Omega^*}
\nc{\s}{\sigma}
\nc{\Si}{\Sigma}
\nc{\ta}{\tau}
\nc{\h}{\theta}
\nc{\z}{\zeta}

\newcommand{\e}{\epsilon}
\newcommand{\opsi}{{\overline{\psi}}}

\nc{\ran}{\rangle}
\nc{\lan}{\langle}

\newcommand{\re}{\operatorname{Re}}
\newcommand{\im}{\operatorname{Im}}
\renewcommand{\Re}{\operatorname{Re}}
\renewcommand{\Im}{\operatorname{Im}}
\newcommand{\ra}{\rightarrow}
\newcommand{\Ran}{\operatorname{Ran}}

\newcommand{\p}{\partial}

\newcommand{\Curl}{\operatorname{curl}}
\newcommand{\Div}{\operatorname{div}}
\newcommand{\curl}{\operatorname{curl}}
\renewcommand{\div}{\operatorname{div}}
\newcommand{\Cov}[1]{\nabla_{\!\!#1}}

\newcommand{\Null}{\operatorname{null}}
\newcommand{\DETAILS}[1]{}

\newcommand{\Lpsi}[2]{\mathscr{L}_{#2}^{}(\tau)}
\newcommand{\Hpsi}[2]{\mathscr{H}_{#2}^{}(\tau)}
\newcommand{\LA}[2]{\vec{\mathscr{L}}_{}^{}(\tau)}
\newcommand{\HA}[2]{\vec{\mathscr{H}}_{}^{}(\tau)}

\numberwithin{theorem}{section}
\numberwithin{equation}{section}
\begin{document}
\date{\DATUM}

\newcommand{\DATUM}{December 4, 2011}              
\pagestyle{myheadings}                         
\markboth{\hfill{AbrikosovLattices, April 24, 2010}}{{AbrikosovLattices, December 4, 2011}\hfill}  %

\title{On Abrikosov Lattice Solutions of the Ginzburg-Landau Equations} 

\author{T. Tzaneteas \thanks{Institut f\"ur Analysis und Algebra, TU Braunschweig, D-38092
Braunschweig, Germany. Supported by DFG Grant No. Ba-1477/5-1.}
\, \qquad and\ \qquad
I. M. Sigal\
\thanks{Dept. of Mathematics, Univ. of Toronto, Toronto,  Canada, M5S 2E4. Supported by NSERC Grant NA7601.}}

\maketitle

  \centerline{\textit{To Misha Shubin, with friendship and admiration}}

\begin{abstract}
Building on earlier work, 
we have given in \cite{TS} 
 a proof of existence of Abrikosov vortex
lattices in the Ginzburg-Landau model of superconductivity and shown that the triangular lattice gives the lowest energy per lattice cell. 
After \cite{TS} was published, we realized that it proves a stronger result than was stated there. This result is recorded in the present paper.
The proofs remain the same as in \cite{TS}, apart from some streamlining.

\medskip

\noindent Keywords: magnetic vortices, superconductivity, Ginzburg-Landau equations, Abrikosov vortex lattices, bifurcations.

\end{abstract}


\section{Introduction}
\label{sec:intro}



\textbf{1.1 The Ginzburg-Landau equations.} The Ginzburg-Landau model of superconductivity describes a superconductor contained in
$\Omega \subset \R^n$, $n = 2$ or $3$, in terms of a complex order parameter $\Psi : \Omega \to \C$, and a
magnetic potential $A : \Omega \to \R^n$\footnote{The Ginzburg-Landau theory is reviewed in every book on superconductivity.
For reviews of rigorous results see the papers \cite{CHO, DGP, GST} and the books \cite{SS, FH, JT, Rub}}. 
The key physical quantities for the model are
\begin{itemize}
    \item the density of superconducting pairs of electrons, $n_s := |\Psi|^2$;
    \item the magnetic field, $B := \Curl A$;
    \item and the current density, $J := \Im(\bar{\Psi}\Cov{A}\Psi)$.
\end{itemize}
In the case $n = 2$, $\Curl A := \frac{\partial A_2}{\partial x_1} - \frac{\partial A_1}{\partial x_2}$ is a scalar-valued function.
The covariant derivative $\Cov{A}$ is defined to be $\nabla - iA$.
The Ginzburg-Landau theory specifies that  the difference between the supperconducting and normal free energies in a state $(\Psi, A)$ is
\begin{equation}
\label{eq:GL-energy}
    E_\Omega(\Psi, A) := \int_\Omega |\Cov{A}\Psi|^2 + |\Curl A|^2 + \frac{\kappa^2}{2} (1 - |\Psi|^2)^2,
\end{equation}
where $\kappa$ is a positive constant that depends on the material properties of the superconductor. (In the problem we consider here it is appropriate
to deal with  Helmholtz free energy at a fixed average magnetic field $b:=\frac{1}{|\Omega|}\int_\Omega \curl{A},$ where $|\Omega|$ is the area or volume of $\Omega$.)
It follows from the Sobolev inequalities that for bounded open sets $\Omega$, $\mathcal{E}_\Omega$ is well-defined and $C^\infty$ as a
functional on the Sobolev space $H^1$. The critical points of this functional must satisfy the well-known Ginzburg-Landau equations
inside $\Omega$:
\begin{subequations}
\label{GLE}
    \begin{equation}
    \label{GLEpsi}
      \Delta_A \Psi = \kappa^2(|\Psi|^2-1)\Psi,
    \end{equation}
    \begin{equation}
    \label{GLEA}
        \Curl^*\Curl A = \Im(\bar{\Psi}\Cov{A}\Psi).
    \end{equation}
\end{subequations}
Here $\Delta_A=- \Cov{A}^*\Cov{A},\ \Cov{A}^*$ and $\Curl^*$ are the adjoints of $\Cov{A}$ and $\Curl$. Explicitly, $\Cov{A}^*F = -\Div F + iA\cdot F$, and
$\Curl^* F = \Curl F$ for $n = 3$ and $\Curl^* f = (\frac{\partial f}{\partial x_2}, -\frac{\partial f}{\partial x_1})$ for $n = 2$.

It is well-known that there exists a
critical value $\kappa_c$ (in the units used here, $\kappa_c = 1/\sqrt{2}$), that separates superconductors into two classes with
different properties: Type I superconductors, which have $\kappa < \kappa_c$ and exhibit first-order phase transitions from the
non-superconducting state to the superconducting state, and Type II superconductors, which have $\kappa > \kappa_c$ and exhibit
second-order phase transitions and the formation of vortex lattices. Existence of these vortex lattice solutions is the subject of the present paper.


\textbf{1.2 Abrikosov lattices.} In 1957, Abrikosov \cite{Abr} discovered solutions of \eqref{GLE} in $n=2$ whose physical characteristics
$n_s$, $B$, and $J$ are periodic with respect to a two-dimensional lattice, while independent of the third dimension, and which have a single
flux per lattice cell\footnote{Such solutions correspond to cylindrical samples. In 2003, Abrikosov received the Nobel Prize for this discovery}.
(In what follows we call such solutions, with $n_s$ and $B$ non-constant, $\Lat$-\textit{lattice solutions}, or, if a lattice $\Lat$ is fixed,
\textit{lattice solutions}. In physics literature they are variously called mixed states, Abrikosov mixed states, Abrikosov vortex
states.) Due to an error of calculation Abrikosov concluded that the lattice which gives the minimum average energy per
lattice cell\footnote{Since for lattice solutions the energy over $\R^2$ (the total energy) is infinite, one considers the average energy per lattice cell,
i.e. energy per lattice cell divided by the area of the  cell.} is the square lattice. Abrikosov's error was corrected by Kleiner, Roth, and Autler \cite{KRA},
who showed that it is in fact the triangular lattice which minimizes the energy.


\textbf{1.3 Results.} In this paper we combine and extend the previous technique to give a complete and self-contained proof of the existence
of Abrikosov lattice solutions. 
As in previous works, we consider only bulk superconductors filling all $\R^3$, with no variation along one direction, so that the problem is reduced to one on $\R^2$.
To formulate our results, for a lattice $\Lat \subset \R^2$, we denote by $\Omega^\Lat$ and $|\Omega^\Lat|$ the basic lattice cell and its area, respectively
(for details see Section \ref{sec:lattice states}).
The flux quantization (see below) implies that
\begin{equation}\label{Ombrel}
	|\Omega^\Lat| = \frac{2\pi n}{b},
\end{equation}
where $b$ is the average magnetic flux per lattice cell, $b := \frac{1}{|\Omega^\Lat|} \int_{\Omega^\Lat} \Curl A$. 
We note that due to the reflection symmetry of the problem we can assume that $b \geq 0$.
\DETAILS{We denote
 the average of a function $f$ over the lattice cell $\Omega^\cL$ of $\cL$ by 
$\langle f \rangle_\cL:=\frac{1}{|\Omega^\cL|}\int_{\Omega^\cL}f.$ 
Next, we introduce the Abrikosov parameter, $\beta(\tau) := \frac{ \lan |\Psi_0|^4 \ran_\cL}{  \lan |\Psi_0|^2 \ran_\cL^2},$
where $\Psi_0$ is the solution of the equation $-\Delta_{A_{0}} \psi_{0} = \kappa^2  \psi_{0}$ on $\Om^\cL$, with the boundary condition  $\Psi_0(x + t) = e^{\frac{i\kappa^2}{2}t\cdot J x}\Psi_0(x)$, for all $t \in \mathcal{L}$. Here $A_{0}:=\frac{\kappa^2}{2} J x$, 
and  $J$ \textbf{is the symplectic matrix}   \begin{equation*}
J := \left( \begin{array}{cc} 0 & -1 \\ 1 & 0 \end{array} \right).\end{equation*}
(We will show below that $\Psi_0$ is unique up to multiplication by a constant and therefore the ratio  $\frac{ \lan |\Psi_0|^4 \ran_\cL}{  \lan |\Psi_0|^2 \ran_\cL^2}$ depends only on the lattice, i.e. $\tau$.) Finally}
     We define
     \begin{equation}\label{kappac}
     \kappa_c(\tau) := \sqrt{\frac{1}{2}\left(1-\frac{1}{\beta (\tau)}\right)} ,
     \end{equation}
where $\beta(\tau)$  is  the Abrikosov parameter, introduced in \eqref{beta} below.  We will prove the following results. 

%

\begin{theorem}\label{thm:main-result}
Let  $\big| b - \kappa^2 \big| \ll 1$ and $(\kappa-\kappa_c(\tau))(\kappa^2-b)\ge 0$. Then for every lattice $\cL$ satisfying 
 \eqref{Ombrel} with $n=1$, the following holds
	\begin{enumerate}[(I)]
	\item 
The equations \eqref{GLE} have an $\Lat$-lattice solution 
in a neighbourhood of the branch of normal solutions.  
	\item The above solution 
is unique, up to symmetry, in a neighbourhood of the  normal branch.
	\item For $(\kappa-\kappa_c(\tau))(\kappa^2-b)\ne 0$, the solution above is real analytic in $b$ in a neighbourhood of $\kappa^2$.
	\item For $\kappa^2 > 1/2$, the lattice shape for which the average energy per lattice cell is minimized approaches the triangular lattice as $b \to \kappa^2$, in the sense that the shape parameter, $\tau_\cL $, of $\cL$ (see Subsection 3.3 below) approaches $ \tau_{triangular} = e^{i\pi/3}$ in $\C$.
	\end{enumerate}	
\end{theorem}

\DETAILS{	\item The condition $\left| |\Omega^\Lat| - \frac{2\pi}{\kappa^2} \right| \ll 1$ says that $b$ is sufficiently close but \textbf{not equal} to 
 the critical value $b_c = \kappa^2$.

\item The solutions we found have  one quantum of flux per cell and the average magnetic flux per cell equal to $b$.}

Since their discovery, Abrikosov lattice solutions have been studied in numerous experimental and theoretical works.
Of more mathematical studies, we mention the articles of Eilenberger \cite{Eil}, Lasher \cite{Lash}, Chapman \cite{Ch} and Ovchinnikov \cite{Ov}.

The rigorous investigation of Abrikosov solutions began soon after their discovery. Odeh \cite{Odeh} sketched a proof of existence for 
various lattices using variational 
and bifurcation techniques.
Barany, Golubitsky, and Turski \cite{BGT} 
applied equivariant bifurcation theory and filled in a number of details, 
and Tak\'{a}\u{c} \cite{Takac} has adapted these results to study the zeros of the bifurcating solutions. Further details and new results, in both, variational and bifurcation, approaches, 
 were provided by \cite{Dutour2, Dutour}.  In particular, \cite{Dutour2} proved partial results on the relation between the bifurcation parameter and the average magnetic field $b$ (left open by previous works) and on the relation between  the Ginzburg-Landau energy and the Abrikosov function, and   \cite{Dutour}  (see also \cite{Dutour2}) found boundaries between superconducting, normal and mixed phases.
\DETAILS{Below, we compare these results with ours:

	\hfill
	\begin{enumerate}[(a)]
	\item \cite{Odeh, Dutour2} showed that for a triangular or rectangular lattice $\cL$ with $b< \kappa^2$ and for $\kappa^2 > \frac{1}{\sqrt{2}}$ an $\Lat$-lattice solution exists as a global minimizer of $\mathcal{E}_{\Omega^\Lat}$.
	\item \cite{Odeh, BGT, Dutour2} laid out main ideas of the proof of 
(I) for $\kappa^2 > \frac{1}{\sqrt{2}}$  and $b< \kappa^2 $.
	\item \cite{Lash, Dutour2} proved partial results on (IV).
	\end{enumerate}} 

Among related results, a relation of the Ginzburg-Landau minimization problem, for a fixed, finite domain and external magnetic field,  
in the regime of $\kappa\ra \infty$, to the Abrikosov lattice variational problem was obtained in \cite{AS, Al2}.

All the rigorous results above deal with Abrikosov lattices with one quantum of magnetic flux per lattice cell. Partial results for higher magnetic fluxes were proven in \cite{Ch, Al}.

 After introducing general properties of \eqref{GLE} in Sections \ref{sec:problem}-\ref{sec:rescaling}, we prove the above theorem in Sections \ref{sec:asymptotics}-\ref{sec:bifurcation-n=1}. 

\noindent \textbf{Acknowledgements}
The second author is grateful to  Yuri Ovchinnikov for many fruitful discussions. A part of this work was done during I.M.S.'s stay at the IAS, Princeton. 


\section{Properties of the  Ginzburg-Landau equations} \label{sec:problem}

\textbf{2.1 Symmetries.} The Ginzburg-Landau equations exhibit a number of symmetries, that is, transformations which map solutions to solutions: 

\noindent The gauge symmetry, 
\begin{equation}\label{eq:gauge-symmetry}
    (\Psi(x), A(x)) \mapsto ( e^{i\eta(x)}\Psi(x),   A(x) + \nabla\eta(x)),\qquad \forall \eta \in C^2(\R^2, \R);
\end{equation}
The 
translation symmetry, 
\begin{equation}\label{eq:translation-symmetry}
    (\Psi(x), A(x)) \mapsto (\Psi(x + t), A(x + t)),\qquad \forall t \in \R^2;
\end{equation}
The rotation and reflection symmetry, 
\begin{equation}\label{eq:rotation-reflection-symmetry}
    (\Psi(x), A(x)) \mapsto (\Psi(R^{-1}x),  RA(R^{-1}x)),\qquad \forall R \in O(2) .
\end{equation}

\medskip

\textbf{2.2 Flux quantization.} One can show that under certain boundary conditions (e.g., `gauge-periodic', see below, or if $\Omega = \R^2$ and $\mathcal{E}_\Omega < \infty$) 
the magnetic flux through $\Omega$ is quantized:
\begin{equation}\label{eq:flux-per-cell}
    \Phi(A) := \int_\Omega \Curl A = 2\pi n
\end{equation}
for some integer $n$.


\textbf{2.3 Elementary solutions.} There are two immediate solutions to the Ginzburg-Landau equations that are homogeneous in $\Psi$. These are the perfect superconductor
solution where $\Psi_S \equiv 1$ and $A_S \equiv 0$, and the \textit{normal} (or non-superconducting) solution where $\Psi_N = 0$ and $A_N$ is such that
$\Curl A_N =: b$ is constant.
(We see that the perfect superconductor is a solution only when 
 $\Phi(A)= 0$. On the other hand, there is a normal solution,  $(\Psi_N = 0,\ A_N,\ \Curl A_N =$ constant), for any condition on $\Phi(A)$.)

Moreover, for any integer $n$ there is a ($n-$) vortex solution of the form
\begin{equation} \label{eq:vort}
   \Psi^{(n)} (x) = f_n (r) e^{in\theta} {\hbox{\quad and \quad}}
   A^{(n)}(x) = a_n (r) \nabla (n\theta) \ ,
\end{equation}
where $(r,\theta)$ are the polar coordinates of $x \in \R^2$, unique up to symmetry transformations (see \cite{BC, GS}). Note that $\Phi(A^{(n)}) = n$.




\section{Lattice states}\label{sec:lattice states}

\textbf{3.1 Periodicity.} Our focus in this paper is on states $(\Psi, A)$ defined on all of $\R^2$, but whose physical properties, the density of superconducting pairs of electrons, $n_s := |\Psi|^2$, the magnetic field, $B := \Curl A$, and the current density, $J := \Im(\bar{\Psi}\Cov{A}\Psi)$, are doubly-periodic with respect
to some lattice $\cL$. We call such states $\cL-$\emph{lattice states}.

One can show that  a state $(\Psi, A) \in H^1_{\textrm{loc}}(\R^2;\C) \times H^1_{\textrm{loc}}(\R^2;\R^2)$ is a $\mathcal{L}$-lattice state if and only if translation by an element of the lattice results in a gauge transformation
    of the state, that is, for each $t \in \mathcal{L}$, there exists a function $g_t \in H^2_{loc}(\R^2;\R)$
    such that $$\Psi(x + t) = e^{ig_t(x)}\Psi(x)\ \mbox{and}\ A(x+t) = A(x) + \nabla g_t(x)$$ almost everywhere.

It is clear that the gauge, translation, and rotation symmetries of the Ginzburg-Landau equations map lattice states to
lattice states. In the case of the gauge and translation symmetries, the lattice with respect to which the solution is
periodic does not change, whereas with the rotation symmetry, the lattice is rotated as well. It is a simple calculation
to verify that the magnetic flux per cell of solutions is also preserved under the action of these symmetries.

Note that $(\Psi, A)$ is defined by its restriction to a single cell and can be reconstructed from this restriction by lattice translations.

\textbf{3.2 Flux quantization.}  The important property of lattice states is that the magnetic flux through a lattice cell is quantized, i.e. \eqref{eq:flux-per-cell} holds,
with $\Omega$ any fundamental cell of the lattice.

Indeed, if $|\Psi| > 0$ on the boundary of the cell, we can write 
$\Psi = |\Psi|e^{i\theta}$ and $0 \leq \theta < 2\pi$. The periodicity of $n_s$ and $J$ ensure the periodicity of $\nabla\theta - A$ and therefore by Green's theorem, $\int_\Omega \Curl A = \oint_{\partial\Omega} A = \oint_{\partial\Omega} \nabla\theta$ and this function is equal to $2\pi n$ since $\Psi$ is single-valued.

Equation \eqref{eq:flux-per-cell} then imposes a condition on the area of a cell, namely, \eqref{Ombrel}.

\textbf{3.3 Lattice shape.} In order to define the shape of a lattice, we identify $x \in \R^2$ with $z = x_1 + ix_2 \in \C$, and view $\mathcal{L}$ as a subset of $\C$. It is a well-known fact (see \cite{Alfors}) that any lattice $\mathcal{L} \subseteq \C$ can be given a basis ${r, r'}$ such
that the ratio $\tau = \frac{r'}{r}$ satisfies the inequalities:
\begin{enumerate}[(i)]
\item $|\tau| \geq 1$;
\item $\Im\tau > 0$;
\item $-\frac{1}{2} < \Re\tau \leq \frac{1}{2}$, and $\Re\tau \geq 0$ if $|\tau| = 1$.
\end{enumerate}
Although the basis is not unique, the value of $\tau$ is, and we will use that as a measure of the shape of the lattice.

Using the rotation symmetry we can assume that $\mathcal{L}$ has as a basis $\{\, re_1, r\tau \,\}$, where $r$ is a positive real number and $e_1 = (1,0)$.

\section{Fixing the gauge and rescaling} \label{sec:rescaling}

In this section we fix the gauge for solutions,  $(\Psi, A)$, of \eqref{GLE} and then rescale them to eliminate the dependence of the size of the lattice on $b$. Our space will then depend only on the number of quanta of flux and the shape of the lattice.

\textbf{4.1 Fixing the gauge.}
The gauge symmetry allows one to fix solutions to be of a desired form.
Let $A^b_0(x) = \frac{b}{2} J x$, 
where  $J$ is the symplectic matrix   \begin{equation*}  J = \left( \begin{array}{cc} 0 & -1 \\ 1 & 0 \end{array} \right).\end{equation*}
     We will use the following preposition, first used by \cite{Odeh} and proved in \cite{Takac} (we provide an alternate proof in Appendix \ref{sec:alternate-proof}).

\begin{proposition}
    \label{thm:fix-gauge}
    Let $(\Psi, A)$ be an $\mathcal{L}$-lattice state, and let $b$ be the average magnetic flux per cell.
    Then there is a $\mathcal{L}$-lattice state $(\phi, A^b_0 + \alpha),$ that is gauge-equivalent to a translation of $(\Psi, A)$, such that
    \begin{enumerate}
    \item[(i)] 
  $\phi(x + t) = e^{\frac{ib}{2}x\cdot J t}\phi(x)$ and   $\alpha(x + t) = \alpha(x)$ for all $t$ in a fixed basis of $\mathcal{L}$;
    \item[(ii)] $\alpha$ has mean zero: $\int_\Omega \alpha = 0$;
    \item[(iii)] $\alpha$ is divergence-free: $\Div \alpha = 0$.
    \end{enumerate}
\end{proposition}


\textbf{4.2 Rescaling.}
Suppose,  that we have a $\mathcal{L}$-lattice state $(\Psi, A)$, where $\mathcal{L}$ has shape $\tau$. Now let $b$ be the average magnetic flux per cell of the state and $n$ the quanta of flux per cell. From the quantization of the flux, we know that
\begin{equation}\label{bOmRel}
    b = \frac{2\pi n}{|\Omega|} = \frac{(r^\tau)^2}{r^2 } n, \quad r^\tau := \left( \frac{2\pi}{\Im\tau} \right)^{\frac{1}{2}},
\end{equation}
We set
$\sigma := \left(\frac{n}{b}\right)^{\frac{1}{2}}=\frac{r}{r^\tau}$. 
We now define the rescaled fields $(\psi, a)$ to be
\begin{equation*}
    (\psi(x), a(x)) := ( \sigma \Psi(\sigma x), \sigma A(\sigma x) ).
\end{equation*}
Let 
$\mathcal{L}^\tau$ be the lattice spanned by $r^\tau$ and $r^\tau\tau$, with $\Omega^\tau$ being a primitive cell of that lattice.
    We note that $|\Omega^\tau| = 2\pi n$.
    We summarize the effects of the rescaling above: 

    \begin{enumerate}[(A)]

    \item $(\psi, a)$ is a $\mathcal{L}^\tau$-lattice state.

    \item   $\frac{1}{|\Omega^\tau|}E_{\Omega^\tau}(\Psi,A) = \mathcal{E}_{\lambda}(\psi,a)$, where $\lambda = \frac{\kappa^2 n}{b}$ and
        \begin{equation}     \label{rEnergy}
            \mathcal{E}_\lambda(\psi, \alpha) = \frac{\kappa^4}{|\Omega^\tau| \lambda^2} \int_{\Omega^\tau}\left( |\Cov{a}\psi|^2 + |\Curl a|^2
                        + \frac{\kappa^2}{2} ( |\psi|^2 - \frac{\lambda}{\kappa^2} )^2\right) \,dx.
        \end{equation}

    \item $\Psi$ and $A$ solve the Ginzburg-Landau equations if and only if $\psi$ and $a$ solve
            \begin{subequations} \label{eq:rGL}
            \begin{equation} \label{rGLpsi}
                (-\Delta_{a}  - \lambda) \psi = -\kappa^2 |\psi|^2\psi,
            \end{equation}
            \begin{equation} \label{rGLA}
                \Curl^*\Curl a = \Im(\bar{\psi}\Cov{a}\psi)
            \end{equation}
        \end{subequations}
 for  $\lambda = \frac{\kappa^2 n}{b}$. The latter equations are valid on  $\Omega^\tau$  with the boundary conditions given in the next statement.

    \item   \label{reduced-gauge-form}
        If $(\Psi, A)$ is of the form described in Proposition \ref{thm:fix-gauge}, then  $\psi$ and $a$ satisfy 
                       \begin{enumerate}[(a)]
            \item \label{reduced-quasiperiodic-bc}  $\psi(x + t) = e^{\frac{in}{2}x\cdot J t}\psi(x)$ and   $\alpha(x + t) = \alpha(x)$,  for $t = r^\tau$, $r^\tau\tau$,  where  
               $ a = A^n_0 + \alpha,\  \mbox{with}\ A^n_0(x) := \frac{n}{2} J x,$ 
            \item $\int_{\Omega^\tau} \alpha = 0$,
            \item $\Div \alpha = 0$.
            \end{enumerate}

    \end{enumerate}

Our problem then is, for each $n = 1,2,\ldots$,  find $(\psi, a)$, 
solving the rescaled Ginzburg-Landau equations \eqref{eq:rGL} and satisfying (iv), and among these find the one that minimizes the average energy $\mathcal{E}_\lambda$.

\section{Asymptotics of solutions to \eqref{eq:rGL}} \label{sec:asymptotics}

In this section we derive properties of families of solutions of \eqref{eq:rGL} depending on $b$, 
\DETAILS{\textbf{We introduce the small parameter $\e:=\sqrt{\kappa^2 n-b}>0$ and display the dependence of solutions on it}. Thus, $(\psi_\e,\ \alpha_\e = A^n_0 + a_\e)$ is a family of solutions of the equations \eqref{eq:rGL} with $\lambda \equiv \lambda_\e= \frac{\kappa^2 n}{b}$ and $\e:=\sqrt{\kappa^2 n-b}>0$.}
  %
 in the regime of $b\ra \kappa^2n$, provided such families exist.
It might be convenient to reparametrize such families  assuming that $b$, or $\lam$, depends on a parameter $\e\ra 0$. Most of the results of this section were first stated in  \cite{Abr} (see also \cite{Ch}). 
In what follows we use the notation
\begin{equation*}
\langle f \rangle:=\frac{1}{|\Omega^\tau|}\int_{\Omega^\tau}f
\end{equation*} for the average of a function $f$ over the lattice cell $\Omega^\tau$.
The main result of this section is the following
\begin{proposition} \label{prop:leadingorder}
Assume the equations \eqref{eq:rGL}
have a family, $(\psi_\e,\ a_\e,\ \lam_\e)$, 
$\e\ra 0$, of  Abrikosov lattice solutions, with $\lan\curl a_\e\ran=n$, 
satisfying
\begin{equation} \label{epsexpansions}
\psi_{\e} = \e\psi_{0}+ O(\epsilon^3),\ a_{\e} =A^n_0 +  \e^2 a_{1} +  O(\epsilon^4), \ \lambda_{\e} = n + \epsilon^2\lambda_1 + O(\epsilon^4),  
\end{equation}
(with the first and second derivatives of the remainders obeying similar estimates). Then $\psi_{0}$, $a_{1}$ and $\lambda_1$ satisfy the equations
\begin{equation} \label{GL:leadingorder} 
-\Delta_{A_{0}^n} \psi_{0} = n  \psi_{0},\ \mbox{and}\ \curl a_{1}=\frac{1}{2}\langle |\psi_{0} |^2 \rangle -\frac{1}{2}|\psi_{0} |^2,\ 
\end{equation}
and 
   \begin{equation} \label{lam1}
\lambda_1=\left[\frac{1}{2}+
+ \left(\kappa^2-\frac{1}{2}\right)\beta(\psi_0)\right]\langle |\psi_{0} |^2 \rangle,
\end{equation}
where 
$\beta(\psi_0)$ is the Abrikosov parameter, defined by
\begin{equation} \label{beta'}
    \beta(\psi_0) := \frac{ \lan |\psi_0|^4 \ran}{  \lan |\psi_0|^2 \ran^2}.
\end{equation}
 Furthermore, the energy is expressed as
               \begin{equation}     \label{Elambda1}
        \mathcal{E}_{\lambda_{\e}}(\psi_{\e}, a_{\e}) = 
        \frac{\kappa^2}{2} +\frac{n^2\kappa^4}{\lambda_\e^2} -  \frac{\kappa^4\e^4}{2n^2 }\left[ (\kappa^2 -  \frac{1}{2})\beta(\psi_0) + \frac{1}{2} \right]\langle |\psi_{0} |^2 \rangle^2 +  O(\e^6).
           \end{equation}
\end{proposition}
\begin{proof} Plugging \eqref{epsexpansions} into \eqref{eq:rGL} 
and taking $\e \rightarrow 0$ gives the first equation in \eqref{GL:leadingorder} and
\begin{equation}
\label{a2-equation}
\curl^*\curl a_{1} =  \Im (\bar{\psi}_{0} \nabla_{A_{0}^n}
\psi_{0} ).
\end{equation}
We show now that
\begin{equation} \label{psi1-current}
\Im (\opsi_{0} \nabla_{A_{0}^n} \psi_{0} ) = -\frac{1}{ 2}
\curl^* |\psi_{0} |^2.
\end{equation}
(Recall, that for a scalar function, $f(x) \in \R,\ \curl^* f = (\p_2 f, -\p_1 f)$ is a vector.)
	It follows from 
\eqref{nullL'}, Section \ref{sec:operators}, that $\psi_{0}$ satisfies the first order equation
\begin{equation} \label{1stordereqnpsi1}
\left((\nabla_{A_{0}^n})_1+i(\nabla_{A_{0}^n})_2\right)\psi_{0} = 0.
\end{equation}
 Multiplying this relation by $\bar{\psi}_{0}$, we obtain $\bar{\psi}_{0}(\nabla_{A_{0}^n})_1\psi_{0}+i\bar{\psi}_{0}(\nabla_{A_{0}^n})_2\psi_{0} =0$. Taking imaginary and real parts of this equation gives
\begin{equation*}\Im \bar{\psi}_{0}(\nabla_{A_{0}^n})_1\psi_{0}=- \Re \bar{\psi}_{0}(\nabla_{A_{0}^n})_2\psi_{0} =-\frac{1}{ 2}\p_{x_2}|\psi_{0}|^2
\end{equation*} and
\begin{equation*}\Im\bar{\psi}_{0}(\nabla_{A_{0}^n})_2\psi_{0} =\Re\bar{\psi}_{0}(\nabla_{A_{0}^n})_1\psi_{0}=\frac{1}{ 2}\p_{x_1}|\psi_{0}|^2,
\end{equation*}
which, in turn, gives \eqref{psi1-current}.

The equations \eqref{a2-equation} and \eqref{psi1-current} give 
$\curl a_1 = H - \frac{1}{2} |\psi_0|^2$,
with $H$ a constant of integration. $H$ has to be chosen so that $\int_{\Omega^\tau} \curl a_1 =0$, which gives  the second equation in \eqref{GL:leadingorder}. 

Now we prove \eqref{lam1}. We multiply the equation \eqref{rGLpsi}
scalarly (in $L^2(\Omega^\tau)$) by $\psi_{0}$, use that the operator $-\Delta_{a}  $ is self-adjoint and $(-\Delta_{a}  - n)\psi_0=0$,
\DETAILS{to obtain
\begin{equation*}
\frac{n-\lambda}{\e^2} \langle \psi_{0},\psi \rangle - 2i\langle \psi_{0}, a\cdot\nabla_{A_0}\psi \rangle 
- \e^2 \langle \psi_{0}, |a|^2\psi \rangle - \kappa\langle\psi_{0}, |\psi|^2\psi \rangle =0.
\end{equation*}}
substitute the expansions \eqref{epsexpansions} and take $\e =0$, to obtain
  \begin{align} \label{solvcond'}
		-\lambda_1 \int_{\Omega^\tau} | \psi_{0} |^2 + 2i\int_{\Omega^\tau} \bar{\psi}_{0} a_{1}\cdot\Cov{A_0^n}\psi_{0} 
+ \kappa^2\int_{\Omega^\tau} |\psi_{0}|^4=0.
	\end{align}
This expression implies that  the imaginary part of  the second term on the left hand side of \eqref{solvcond'} is zero. (We arrive at the same conclusion by integrating by parts and using that $\div a_{1}=0$.) 
   Therefore
    \begin{align*}
        2i \int_{\Omega^\tau} \bar{\psi}_{0}  a_{1} \cdot\Cov{A_0^n}\psi_{0}
                    &= -2 \int_{\Omega^\tau} a_{1} \cdot \Im( \bar{\psi}_{0} \Cov{A_0^n}\psi_{0}  )  = -2 \int_{\Omega^\tau} a_{1} \cdot \Curl^* \Curl a_{1}.  
         \end{align*}
Integrating the last term by parts, we obtain $2i \int_{\Omega^\tau} \bar{\psi}_{0}  a_{1} \cdot\Cov{A_0^n}\psi_{0}= -2 \int_{\Omega^\tau} (\Curl a_{1} )^2  .$
 Using this equation and the second equation in \eqref{GL:leadingorder}, we obtain
 \begin{align} \label{intaimpsinablapsi}
		 2i \int_{\Omega^\tau} \bar{\psi}_{0}  a_{1} \cdot\Cov{A_0^n}\psi_{0}
		 = - \frac{1}{2} \int_{\Omega^\tau} |\psi_{0}|^4 \, + \frac{1}{2}\langle |\psi_{0} |^2 \rangle\int_{\Omega^\tau} |\psi_{0} |^2 .
    \end{align}
This equation together with \eqref{solvcond'} and the definition \eqref{beta'} gives \eqref{lam1}.

Now, we prove statement \eqref{Elambda1} 
about the  Ginzburg-Landau energy.
   Multiplying \eqref{rGLpsi} scalarly by $\psi$ and integrating by parts gives
   \begin{equation*}
  		\int_{\Omega^\tau} |\nabla_a \psi|^2 = \kappa^2 \int_{\Omega^\tau} \left(\lambda|\psi|^2 - \kappa^2|\psi|^4\right).
	\end{equation*}
Substituting this into the expression for the energy, 
 we find
   \begin{equation}   \label{asymp:Elambda'}
		\mathcal{E}_{\lambda}(\psi, a)  = \frac{\kappa^4}{\lambda^2} \big( \frac{\lambda^2}{2\kappa^2} - \frac{\kappa^2}{2} \lan|\psi|^4\ran + \lan|\Curl a|^2\rangle\big).
            \end{equation}
    Using the expansions  \eqref{epsexpansions} and the facts that $\Curl A_0^n = n$ and $\langle\Curl a_1\rangle= 0$  gives
     \begin{equation}
    \label{asymp:Elambda''}
	    \mathcal{E}_{\lambda_{\e}}(\psi_{\e}, a_{\e})  = \frac{\kappa^2}{2} +\frac{n^2\kappa^4}{\lambda_\e^2} +  \frac{\kappa^4 \e^4}{\lambda_\e^2}   \left(  - \frac{\kappa^2}{2} \langle |\psi_0|^4 \rangle+ \langle|\Curl a_1|^2\rangle \right)
	    		+ O(\e^6).
            \end{equation}
Next, using the second equation in \eqref{GL:leadingorder} 
together with  
the last equation in \eqref{epsexpansions} and the definition \eqref{beta'}
yields \eqref{Elambda1}.
\end{proof}
Eqn \eqref{lam1} fixes the parameter $\e$ uniquely up to the normalization of $\psi_0$. Indeed, we observe that the third equation in \eqref{epsexpansions} implies   $\e^2 = \frac{\lambda- n}{\lambda_1}+O((\lambda- n)^2)$, which, together with the definition $\lambda = \frac{\kappa^2 n}{b}$ and  \eqref{lam1}, yields
\begin{equation}    \label{eps}
         \e^2= \frac{n(\kappa^2 - b)}{\kappa^2 [(\kappa^2 -  \frac{1}{2})\beta(\psi_0) +\frac{1}{2}]\lan |\psi_0|^2\ran} +O((\kappa^2 - b)^2).
    \end{equation}
This equation implies the following necessary condition on existence of the solutions: \begin{equation}  \label{nec-cond}
b\le \kappa^2$ if $(\kappa^2 -  \frac{1}{2})\beta(\psi_0) +\frac{1}{2}\ge 0\ \quad \textrm{and}\ \quad b>\kappa^2\ \textrm{if}\ (\kappa^2 -  \frac{1}{2})\beta(\psi_0) +\frac{1}{2}<0.
\end{equation}
\eqref{eps} together with \eqref{Elambda1} yields
    \begin{equation}     \label{Elam2}
        \mathcal{E}_{\lambda_{\e}}(\psi_{\e}, a_{\e}) = \frac{\kappa^2}{2} +
         b-    \frac{(\kappa^2 - b)^2}{ (2\kappa^2 -  1)\beta(\psi_0) +1}   +  O((\kappa^2 - b)^3).
           \end{equation}


\section{The linear problem} 
\label{sec:operators}


In this section we solve the linear problem:
$-\Delta_{A_{0}^n} \psi_{0} = n  \psi_{0},$ for $\psi$ satisfying the gauge - periodic boundary condition
$\psi(x + t) = e^{\frac{in}{2}x\cdot Jt}\psi(x)$,  for $t = r^\tau$, $r^\tau\tau$,
(see  \eqref{GL:leadingorder} and \eqref{reduced-gauge-form}.)
We introduce the harmonic oscillator annihilation and creation  operators, $\alpha^n$ and $(\alpha^n)^*$, with 
    \begin{equation}
        \alpha^n := (\nabla_{A_0^n})_1 + i(\nabla_{A_0^n})_2 =\partial_{x_1} + i\partial_{x_2} + \frac{n}{2}x_1 + \frac{in}{2}x_2.
    \end{equation}
    One can verify that these operators satisfy the following relations:
    \begin{enumerate}
    \item $[\alpha^n, (\alpha^n)^*] = 2\Curl A_0^n =2n$;
    \item $-\Delta_{A_{0}^n} - n = (\alpha^n)^*\alpha^n$.
    \end{enumerate}
    As for the harmonic oscillator (see for example \cite{GS2}), this gives explicit information about $\sigma(-\Delta_{A_{0}^n})$: 
	\begin{equation}
	\sigma(-\Delta_{A_{0}^n}) = \{\, (2k + 1)n : k = 0, 1, 2, \ldots \,\},
	\end{equation}
	and each eigenvalue is of the same multiplicity.
    Furthermore, the above properties imply
    \begin{equation} \label{nullL'}
        \Null (-\Delta_{A_{0}^n} - n) = \Null \alpha^n.
    \end{equation}

    We can now prove the following.
    \begin{proposition}\label{prop:nullspace} $ \Null (-\Delta_{A_{0}^n} - n)$ is given by 
     \begin{equation} \label{nullL} \Null (-\Delta_{A_{0}^n} - n) =
     		\{\, e^{\frac{in}{2}x_2(x_1 + ix_2) }\sum_{k=-\infty}^{\infty} c_k e^{ik\sqrt{2\pi \Im \tau}(x_1 + ix_2)}\ |\ c_{k + n} = e^{in\pi\tau} e^{i2k\pi\tau} c_k \}
        \end{equation}
     and therefore, in particular,
        $\dim_\C \Null (-\Delta_{A_{0}^n} -n)= n$.
    \end{proposition}
    \begin{proof} We find $\Null \alpha^n$.
        A simple calculation gives the following operator equation
        \begin{equation*}
            e^{\frac{n}{4}|x|^2}\alpha^ne^{-\frac{n}{4}|x|^2} = \partial_{x_1} + i\partial_{x_2}.
        \end{equation*}
        This immediately proves that $\psi \in \Null \alpha^n$ if and only if $\xi = e^{\frac{n}{4}|x|^2}\psi$ satisfies $\partial_{x_1}\xi + i\partial_{x_2}\xi = 0$.
        We now identify $x \in \R^2$ with $z = x_1 + ix_2 \in \C$ and see that this means that $\xi$ is analytic. We therefore define the entire function
        $\Theta$ to be
        \begin{equation*}
            \Theta(z) = e^{ -\frac{n z^2}{2\pi\Im\tau} } \xi \left( \sqrt{\frac{2}{\pi\Im\tau}} z \right)= e^{ \frac{n (|z|^2-z^2)}{2\pi\Im\tau} } \psi \left( \sqrt{\frac{2}{\pi\Im\tau}} z \right).
        \end{equation*}
The quasiperiodicity of $\psi$ transfers to $\Theta$ as follows
        \begin{subequations}
            \begin{equation*}
                \Theta(z + \pi) = \Theta(z),
            \end{equation*}
            \begin{equation*}
                \Theta(z + \pi\tau) =  e^{ -2inz } e^{ -in\pi\tau  } \Theta(z).
            \end{equation*}
        \end{subequations}

        To complete the proof, we now need to show that the space of the analytic functions which satisfy these relations form a vector space
        of dimension $n$. It is easy to verify that the first relation ensures that $\Theta$ have a absolutely convergent Fourier expansion of the form
        \begin{align*}
            \Theta(z) = \sum_{k=-\infty}^{\infty} c_k e^{2kiz}.
        \end{align*}
        The second relation, on the other hand, leads to relation for the coefficients of the expansion. Namely, we have
           $ c_{k + n} = e^{in\pi\tau} e^{2ki\pi\tau} c_k$ 
        and that means such functions are determined solely by the values of $c_0,\ldots,c_{n-1}$ and therefore form an $n$-dimensional vector space.
    \end{proof}

\section{Reformulation of the problem} \label{sec:reformulation}

In this section we reduce two equations \eqref{eq:rGL} for $\psi$ and $a$ to a single equation for $\psi$. We introduce  the spaces  $\Lpsi{2}{n}:=L^2(\Omega^\tau, \C)$ and $\LA{p}{\Div,0}:=\{a\in L^2(\Omega^\tau, \R^2)\ |\ \lan a\ran = 0,\ \div a=0$, in the distributional sense$\}$, with the domains being  the Sobolev space of order two, $\Hpsi{2}{n}$, whose elements satisfy the quasiperiodic boundary condition (d) in part \eqref{reduced-gauge-form} above, and the Sobolev space of order two, $\HA{2}{\Div,0}$, whose elements satisfy the periodic boundary conditions with respect to $\mathcal{L}^\tau$, have mean zero, and are divergence free.
Substituting $a = A^n_0 + \alpha$, we rewrite  \eqref{eq:rGL} as
\begin{subequations} \label{psiaeqs}
    \begin{equation}
    \label{psieq} 
        (L^n - \lambda)\psi + 2i\alpha\cdot\nabla_{A^n_0}\psi 
        + |\alpha|^2\psi + \kappa^2|\psi|^2\psi =0,
    \end{equation}
    \begin{equation}
    \label{aeq} 
        (M + |\psi|^2)\alpha - \Im(\bar{\psi}\Cov{A^n_0}\psi )=0,
    \end{equation}
\end{subequations}
where
\begin{equation}\label{OP:LN}
    L^n := -\Delta_{A^n_0}
\mbox{ and }
    M := \Curl^*\Curl,
\end{equation}
\DETAILS{Equations \eqref{eq:rGL} is then equivalent to
\begin{subequations}
\label{eq:F-i=0}
    \begin{equation}
    \label{eq:F-1=0}
        F_1(\lambda, \psi, a) = 0,
    \end{equation}
    \begin{equation}
    \label{eq:F-2=0}
        F_2(\psi, a) = 0.
    \end{equation}
\end{subequations}

We summarize the properties of these two maps in the following proposition, whose straightforward proof is omitted.
\begin{proposition}
    \hfill
    \begin{enumerate}[(a)]
    \item $F_1$ and $F_2$ are $C^\infty$,
    \item for all $\lambda$, $F_1(\lambda, 0, 0) = 0$ and $F_2(0,0) = 0$,
    \item for all $\alpha \in \R$, $F_1(\lambda, e^{i\alpha}\psi, a) = e^{i\alpha}F_1(\lambda, \psi, a)$ and
            $F_2(e^{i\alpha}\psi, a) = F_2(\psi, a)$.
    \end{enumerate}
\end{proposition}}
%
defined on the spaces  $\Lpsi{2}{n}$ and $\LA{p}{\Div,0}$. Their properties that will be used below are summarized in the following propositions: 

\begin{proposition}\label{thm:L-spec}
    $L^n$ is a self-adjoint operator on $\Hpsi{2}{n}$ with spectrum $\sigma(L^n) = \{\, (2k + 1)n : k = 0, 1, 2, \ldots \,\}$ and
    $\dim_\C \Null (L^n - n) = n$.
\end{proposition}

\begin{proposition}\label{thm:M-spec}
    $M$ is a strictly positive operator on $\HA{2}{\Div,0}$ with discrete spectrum.
\end{proposition}

The proofs of these results are standard and, for the convenience of the reader, are given below.
\begin{proof}[Proof of Proposition \ref{thm:M-spec}]
    The fact that $M$ is positive follows immediately from its definition. We note that its being strictly positive is the result of
    restricting its domain to elements having mean zero.
\end{proof}

\begin{proof}[Proof of Proposition \ref{thm:L-spec}]
    First, we note that  $L^n$ is clearly a positive self-adjoint operator. To see that it has discrete spectrum, we first note that the inclusion
    $H^2 \hookrightarrow L^2$ is compact for bounded domains in $\R^2$ with Lipschitz boundary (which certainly includes lattice cells). Then for any
    $z$ in the resolvent set of $L^n$, $(L^n - z)^{-1} : L^2 \to H^2$ is bounded and therefore $(L^n - z)^{-1} : L^2 \to L^2$ is compact. In fact, the spectrum of  $L^n$ was found explicitly in the previous section. This completes the proof of Proposition \ref{thm:L-spec}.
    \end{proof}

We first solve the second equation \eqref{aeq} for $\alpha$ in terms of $\psi$, using the fact that $M$ is a strictly positive operator, and that
  $\div J_\alpha=0,\ \lan J_\alpha\ran =0$, where $ J_\alpha:=\Im\{ \bar{\psi}\Cov{A^n_0+\alpha}\psi \}$. The last two relations follow for any solution $(\psi, a)$ of \eqref{eq:rGL} by differentiating the equation $\cE_\lam(e^{s\chi}\psi, a+s\nabla\chi)=\cE_\lam(\psi, a) $, w.r.to $s$ at $s=0$, which gives
$ \p_\psi \cE_\lam(\psi, \alpha)i\chi\psi+\p_\alpha \cE_\lam(\psi, \alpha)\nabla\chi=0$. Here $\p_\psi \cE_\lam(\psi, \alpha)$ and $\p_\alpha \cE_\lam(\psi, \alpha)$ are the G\^ateaux derivatives of $\cE_\lam(\psi, a)$ w.r.to $\psi$ and $a$. Since $ \p_\psi \cE_\lam(\psi, a)=0$, this yields
$ 0=  \int_{\Omega^\tau} (M -J_\alpha)\cdot\nabla\chi=  \int_{\Omega^\tau} \div J_\alpha\chi.$
Since the last equation holds for any $\chi \in H_1(\Omega^\tau, \R)$, we conclude that $\div J_\alpha=0$.
Choosing $\chi  = h\cdot x,\ \forall h\in \R^2$, in the equation $ 0=  \int_{\Omega^\tau} (M -J_\alpha)\cdot\nabla\chi$, we find $\lan J_\alpha\ran =0$.
Now, \eqref{aeq} can be rewritten as a fixed point problem $a = M^{-1}J_\alpha$, which has a unique solution in $\HA{2}{\Div,0}$. The latter can be rewritten as 
$\alpha=\alpha(\psi)$, where
\begin{equation} \label{eq:apsi}
    \alpha(\psi) = (M + |\psi|^2)^{-1}\Im(\bar{\psi}\Cov{A^n_0}\psi).
\end{equation}
We collect the elementary properties of the map $\alpha$ in the following proposition, where we identify $\Hpsi{2}{n}$ with a real Banach space using $\psi \leftrightarrow \overrightarrow{\psi}:=(\Re \psi, \Im \psi)$.
\begin{proposition}
 The unique solution, $\alpha(\psi)$, of \eqref{aeq} maps $\Hpsi{2}{n}$ to $\HA{2}{\Div,0}$ and
 has the following properties:
    \begin{enumerate}[(a)]
    \item $\alpha(\cdot)$ is analytic as a map between real Banach spaces.
    \item $\alpha(0) = 0$.
    \item For any $\delta \in \R$, $\alpha(e^{i\delta}\psi) = \alpha(\psi)$.
    \end{enumerate}
\end{proposition}
\begin{proof}
    The only statement that does not follow immediately from the definition of $\alpha$ is (a). It is clear that $\Im(\bar{\psi}\Cov{A^n_0}\psi)$ is
    real-analytic as it is a polynomial in $\psi$ and $\nabla\psi$, and their complex conjugates.   We also note that $(M - z)^{-1}$ is
    complex-analytic in $z$ on the resolvent set of $M$, and therefore, $(M + |\psi|^2)^{-1}$ is analytic. (a) now follows.
\end{proof}

Now we substitute the expression \eqref{eq:apsi} for $\alpha$ into \eqref{psieq} to get a single equation
\begin{equation}\label{Feq0} F(\lambda, \psi) = 0,
\end{equation}
 where the map
$F : \R \times \Hpsi{2}{n} \to \Lpsi{2}{n}$ is defined as 
\begin{equation}\label{F}
    F(\lambda, \psi) = (L^n - \lambda)\psi + 2i\alpha(\psi)\cdot\nabla_{A^n_0}\psi + |\alpha(\psi)|^2\psi + \kappa^2|\psi|^2\psi.
\end{equation}

For a map $F(\psi)$, we denote by $\p_\psi F(\phi)$ its G\^ateaux derivative in $\psi$ at $\phi$. The following proposition lists some properties of $F$.
\begin{proposition}\label{prop:F-prop}
    \hfill
    \begin{enumerate}[(a)]
    \item $F$ is analytic as a map between real Banach spaces,
    \item for all $\lambda$, $F(\lambda, 0) = 0$,
    \item for all $\lambda$, $\p_\psi F(\lambda, 0)=L^n - \lambda, $
    \item for all $\delta \in \R$, $F(\lambda, e^{i\delta}\psi) = e^{i\delta}F(\lambda, \psi)$.
	\item \label{psiF-real} for all $\psi$, $\langle \psi, F(\lambda, \psi) \rangle \in \R$.
    \end{enumerate}
\end{proposition}
\begin{proof}
    The first property follows from the definition of $F$ and the corresponding analyticity of $a(\psi)$. (b) through (d) are straightforward calculations. For (e), we calculate that
   	\begin{align*}
		\langle \psi, F(\lambda, \psi) \rangle
		&= \langle \psi, (L^n - \lambda)\psi \rangle + 2i\int_{\Omega^\tau} \bar{\psi}\alpha(\psi)\cdot\nabla\psi\\
			&+ 2\int_{\Omega^\tau} (\alpha(\psi)\cdot A_0^n)|\psi|^2
			+ \int_{\Omega^\tau} |\alpha(\psi)|^2 |\psi|^2
			+ \kappa^2 \int_{\Omega^\tau} |\psi|^4.
	\end{align*}
	The final three terms are clearly real and so is the first because $L^n - \lambda$ is self-adjoint. For the second term we calculate the complex conjugate and see that
	\begin{equation*}
		\overline{ 2i\int_{\Omega^\tau} \bar{\psi}\alpha(\psi)\cdot\nabla\psi }
			= -2i\int_{\Omega^\tau} \psi \alpha(\psi)\cdot\nabla\bar{\psi}
			= 2i\int_{\Omega^\tau} (\nabla\psi \cdot \alpha(\psi))\bar{\psi},
	\end{equation*}
	where we have integrated by parts and used the fact that the boundary terms vanish due to the periodicity of the integrand and that $\Div \alpha(\psi) = 0$. Thus this term is also real and (e) is established.
\end{proof}

\section{Reduction to a finite-dimensional problem} \label{sec:reduction}

In this section we reduce the problem of solving  the equation $F(\lambda,\psi) = 0$ to a finite dimensional problem. We address the latter in the next section.
We use the standard method of Lyapunov-Schmidt reduction. Let 
   $X:=\Hpsi{2}{n}$ and $Y:= \Lpsi{2}{n}$ and let $K = \Null (L^n - n)$. We let $P$ be the Riesz projection onto $K$, that is,
    \begin{equation}
        P := -\frac{1}{2\pi i} \oint_\gamma (L^n - z)^{-1} \,dz,
    \end{equation}
    where $\gamma \subseteq \C$ is a contour around $n$ that contains no other points of the spectrum of $L^n$.
    This is possible since $n$ is an isolated eigenvalue of $L^n$. $P$ is a bounded, orthogonal projection, and if we
    let $Z := \Null P$, then $Y = K \oplus Z$. We also let $Q := I - P$, and so $Q$ is a projection onto $Z$.

    The equation $F(\lambda,\psi) = 0$ is therefore equivalent to the pair of equations
    \begin{align}
        \label{BT:eqn1} &P F(\lambda, P\psi + Q\psi) = 0, \\
        \label{BT:eqn2} &Q F(\lambda, P\psi + Q\psi) = 0.
    \end{align}

    We will now solve \eqref{BT:eqn2} for $w = Q\psi$ in terms of $\lambda$ and $v = P\psi$. To do this, we introduce the map
    $G : \R \times K \times Z \to Z$ to be $G(\lambda, v, w) := QF(\lambda, v + w)$.    Applying the Implicit Function Theorem
    to $G$, we obtain a function $w : \R \times K \to Z$, defined on a neighbourhood of $(n, 0)$,
    such that $w = w(\lambda, v)$ is a unique solution to $G(\lambda, v, w) = 0$, for $(\lambda, v)$ in that neighbourhood. This solution has the following properties
  \begin{equation}\label{w-prop1}
  w(\lambda, v)  \ \mbox{real-analytic in}\    (\lambda, v);  \end{equation}
     \begin{equation} \label{w-prop2}
      w(\lambda, v)=O(|v|^2)\ \mbox{and}\      \p_\lam w(\lambda, v)=O(|v|^2).\end{equation}
The last property follows from the fact that the last three terms in \eqref{F} are at least quadratic in $\psi =v+w$.

    We substitute the solution $w = w(\lambda, v)$ into \eqref{BT:eqn1} and see that the latter equation 
     in a neighbourhood of $(n, 0)$ is equivalent to the equation  (the \emph{bifurcation equation})
           \begin{equation}  \label{bif-eqn}
        \gamma(\lambda, v):= PF(\lambda, v + w(\lambda, v)) = 0.
    \end{equation}
    Note that 
     $\gamma : \R \times K \to \C$. 
       We have shown that in a neighbourhood of $(n, 0)$ in $\R \times X$, $(\lambda, \psi)$ solves $F(\lambda, \psi) = 0$
    if and only if $(\lambda, v)$, with $v = P\psi$, solves
    \eqref{bif-eqn}. Moreover, the solution $\psi$ of $F(\lambda, \psi) = 0$ can be reconstructed from the solution $v$ of \eqref{bif-eqn} according to the formula
     \begin{equation} \label{psivw}
     \psi =v+ w(\lambda, v).
    \end{equation} 
 Finally we note that $w$ and $\gamma$ inherit the symmetry of the original equation:
    \begin{lemma}\label{gam-gaugeinv}
        For every $\delta \in \R$, $w(\lambda, e^{i\delta}v) = e^{i\delta} w(\lambda, v)$ and $\gamma(\lambda, e^{i\delta}v) = e^{i\delta} \gamma(\lambda, v)$.
    \end{lemma}
    \begin{proof}
        We first check that $w(\lambda, e^{i\delta}v) = e^{i\delta} w(\lambda, v)$. We note that by definition of $w$,
         \begin{align*}G(\lambda,  e^{i\delta} v, w(\lambda, e^{i\delta} v)) = 0, \end{align*} but by the symmetry of $F$, we also have
        $G(\lambda,  e^{i\delta} v,  e^{i\delta} w(\lambda, v)) =  e^{i\delta} G(\lambda,v, w(\lambda,v)) = 0$. The uniqueness of $w$
        then implies that $w(\lambda, e^{i\delta} v) = e^{i\delta} w(\lambda, v)$.
        We can now verify that
        \begin{align*}
            &\gamma(\lambda, e^{i\delta} v) = PF(\lambda, e^{i\delta} v + w(\lambda,  e^{i\delta} v))\\
                &= e^{i\delta} PF(\lambda, v + w(\lambda, v) ) \rangle = e^{i\delta}\gamma(\lambda,v).
        \qedhere
        \end{align*}
    \end{proof}

Solving the bifurcation equation \eqref{bif-eqn} is a subtle problem, unless $n=1$. In the latter case, this is done in the next section.

We conclude this section with mentioning an approach to finding solutions to  the bifurcation equation \eqref{bif-eqn} for any  $n$. For a fixed $n$, we define the first reduced energy $E_{\lambda}(\psi) := \mathcal{E}_{\lambda}(\psi, a)$, where $a = A^n_0 + \alpha,$ with $A^n_0(x) := \frac{n}{2} J x$ and
$\alpha(\psi) = (M + |\psi|^2)^{-1}\Im(\bar{\psi}\Cov{A^n_0}\psi)$ (see \eqref{reduced-gauge-form} 
and \eqref{eq:apsi}). Critical points of this energy solve the equation $F(\lambda, \psi) = 0$.

    Next, we introduce the finite dimensional effective Ginzburg-Landau energy
\begin{align*} e_{\lambda}(v) := E_{\lambda}(v+ w(\lambda, v)).\end{align*}
It is a straightforward to show that 

    (i) $e_{\lambda}(v)$ has a critical point $v_0$ iff  $E_{\lambda}(u)$ has a critical point $u_0=v_0+ w(\lambda, v_0)$;

    (ii) Critical points, $v_0$, of  $e_{\lambda}(v)$ solve the equation \eqref{bif-eqn};

       (iii) $e_{\lambda}(v)$ is gauge invariant, $e_{\lambda}(e^{i\delta}v) = e_{\lambda}(v).$

\noindent One can use $e_{\lambda}(v)$ to investigate 
solutions of the equation \eqref{bif-eqn} for any  $n$. 

\section{Proof of Theorem \ref{thm:main-result}}
 \label{sec:bifurcation-n=1}

In this section we look at the case $n = 1$, and look for solutions near the trivial solution. For convenience we \textit{drop the (super)index} $n = 1$ from the notation.
Recall that $\psi_0$ is a non-zero element in the nullspace of the operator $L^n - n$ acting on $\Hpsi{2}{n}$. Since by Proposition \ref{prop:nullspace},
this nullspace is a one-dimensional complex
subspace for $n=1$, the Abrikosov function, $\beta(\psi_0)$, defined in \eqref{beta'}, depends only on $\tau$.
Therefore we write $\beta(\tau)\equiv \beta(\psi_0)$, so that
\begin{equation} \label{beta}
    \beta(\tau) := \frac{ \lan |\psi_0|^4\ran }{ \lan |\psi_0|^2 \ran^2 }.
\end{equation}
We begin with the following result which gives the existence and uniqueness of the Abrikosov lattices.

\begin{theorem} \label{thm:bifurcation-resultN1}
    For every $\tau$ 
    there exist $\epsilon > 0$ and a branch, $ (\lambda_s, \psi_s, \alpha_s)$, $s \in [0, \sqrt{\e})$, 
    of nontrivial solutions of  the rescaled Ginzburg-Landau equations \eqref{eq:rGL}, unique modulo the global gauge symmetry  (apart from the trivial solution $(1,0,A_0)$) in a sufficiently small neighbourhood of $(1,0,A_0)$ in $\R \times \mathscr{H}(\tau) \times \vec{\mathscr{H}}(\tau)$, and such that
    \begin{equation}\label{sexpansions}
    \begin{cases}
        \lambda_s = 1 + g_\lambda(s^2), \\
        \psi_s = s\psi_0 + sg_\psi(s^2), \\
        a_s = A_0 + g_a(s^2),
    \end{cases}
    \end{equation}
    where $(L - 1)\psi_0 = 0,$ 
    $ g_\psi$ is orthogonal to $\Null(L - 1)$,  $g_\lambda : [0,\epsilon) \to \R$, $g_\psi : [0,\epsilon) \to \mathscr{H}(\tau)$, and $g_\al : [0,\epsilon) \to \vec{\mathscr{H}}(\tau)$ are real-analytic functions such that $g_\lambda(0) = 0$, $g_\psi(0) = 0$, $g_\al(0) = 0$ and  
        \begin{align} \label{glambda'}
g'_\lambda(0) = \left[\left(\kappa^2 - \frac{1}{2} \right) 
\beta (\tau)+\frac{1}{2}\right]\lan |\psi_0|^2\ran.
	\end{align}
    \end{theorem}
\begin{proof} The proof of this theorem is a slight modification of a standard result from bifurcation theory. Our goal is to solve the equation \eqref{bif-eqn} for $\lam$. Since the projection $P$, defined there, is rank one and self-adjoint, we have
    \begin{equation}    \label{BT:P}
        P\psi = \frac{1}{\|\psi_0\|^2} \langle \psi_0, \psi \rangle \psi_0,\ \textrm{with}\ \psi_0 \in \Null \p_\psi F(\lambda_0, 0).
    \end{equation}
	We can therefore view the function $\gamma$ in the bifurcation equation \eqref{bif-eqn} as a map $\gamma : \R \times \C \to \C$, where
	\begin{equation} \label{gam}
		\gamma(\lambda, s) = \langle \psi_0, F(\lambda, s\psi_0 + w(\lambda, s\psi_0) \rangle.
	\end{equation}
We now show that $\gamma(\lambda, s) \in \R$.
Since the projection $Q$ is self-adjoint,  $Qw(\lambda, v) = w(\lambda, v), $ $w(\lambda, v)$ solves  $ QF(\lambda, v + w)=0$ and $v=s\psi_0,$ we have
	\begin{align*}
		\langle w(\lambda, s\psi_0), F(\lambda, s\psi_0+ w(\lambda, s\psi_0)) \rangle
		= \langle w(\lambda, s\psi_0), QF(\lambda, s\psi_0 + w(\lambda, s\psi_0)) \rangle
		= 0.
	\end{align*}
	Therefore, for $s \neq 0$,
	\begin{equation*}
		\langle \psi_0, F(\lambda, s\psi_0 + \Phi(\lambda, s\psi_0)) \rangle
		= s^{-1}\langle s\psi_0 + w(\lambda, s\psi_0), F(\lambda, s\psi_0 + w(\lambda, s\psi_0)) \rangle,
	\end{equation*}
	and this is real by property \eqref{psiF-real} of Proposition \ref{prop:F-prop}. Thus,
since by Lemma \ref{gam-gaugeinv}, $\gamma(\lambda, s) = e^{i\arg s} \gamma(\lambda, |s|)$, it therefore  suffices to solve the equation
\begin{equation} \label{bif-eqn'} \gamma_0(\lambda, s) =0
\end{equation}
for the restriction $\gamma_0 : \R \times \R \to \R$ of the function $\gamma$ to $\R \times \R$, i.e. for real $s$. Since by \eqref{w-prop2}, $w(\lambda, s\psi_0)=O(s^2)$ and therefore \eqref{bif-eqn'} has the trivial branch of solutions $s\equiv 0$ for all $\lam$. Hence we factorize $\gamma_0(\lambda, s)$ as $\gamma_0(\lambda, s) =s\gamma_1(\lambda, s)$, i.e. 
 we define the function
 \begin{align} \label{gam1}
		\gamma_1(\lambda, s) &:= s^{-1}\gamma_0(\lambda, s) 
= \langle \psi_0, F(\lambda, \psi_0 + s^{-1}w(\lambda, s\psi_0) \rangle
	\end{align}
and solve the equation $\gamma_1(\lambda, s) =0$. The definition of the function $\gamma_1(\lambda, s)$ implies that it has the following properties: $\gamma_1(\lambda, s)$ is real-analytic, $\gamma_1(\lambda, -s) =-\gamma_1(\lambda, s)$, $\gamma_1 (1, 0)=0$ and, by \eqref{F} and \eqref{w-prop2}, $\p_\lam\gamma_1(1, 0)=-\|\psi_0\|^2\ne 0$. Hence by a standard application of  the Implicit Function Theorem,  there is $\epsilon > 0$ and a real-analytic function
$\widetilde{\phi}_\lambda : (-\epsilon, \epsilon) \to \R$ such that $\widetilde{\phi}_\lambda(0) = \lambda_0$ and $\gamma(\lambda, s) = 0$ with $|s| < \epsilon$ if and only if either $s = 0$ or $\lambda = \widetilde\phi_\lambda(s)$.     

We also note that by the symmetry, $\widetilde{\phi}_\lambda(-s) = \widetilde{\phi}_\lambda(|s|) = \widetilde{\phi}_\lambda(s)$, so $\widetilde{\phi}_\lambda$ is an even real-analytic function, and therefore must in fact be a function solely of $|s|^2$. We therefore set $\phi_\lambda(s) = \widetilde{\phi}_\lambda(\sqrt{s})$, and
so $\phi_\lambda$ is real-analytic.

We now define $\phi_\psi : (-\epsilon, \epsilon) \to \R$ to be
    \begin{equation}
        \phi_\psi(s) = \begin{cases}
                        s^{-1} w(\phi_\lambda(s), t\psi_0) &  s \neq 0, \\
                        0 & s = 0,
                    \end{cases}
    \end{equation}
$\phi_\psi$ is also real-analytic and satisfies $s\phi_\psi(s^2) = w(\phi_\lambda(s^2), s\psi_0)$ for any $s \in [0, \sqrt{\e})$. 

Now, we know that there is a neighbourhood of $(\lambda_0, 0)$ in $\R \times \Null \p_\psi F(\lambda_0, 0)$ such that in that neighbourhood $F(\lambda, \psi) = 0$
if and only if $\gamma(\lambda, s) = 0$ where $P\psi = s\psi_0$. By taking a smaller neighbourhood if necessary, we have proven that
$F(\lambda, \psi) = 0$ in that neighbourhood if and only if either $s = 0$ or $\lambda = \phi_\lambda(s^2)$. If $s = 0$, we have
$\psi = s\psi_0 + s\phi_\psi(s^2) = 0$ which gives the trivial solution. In the other case, $\psi = s\psi_0 + s\phi_\psi(s^2)$. 

The above gives us a neighbourhood of $(1, 0)$ in $\R \times \mathscr{H}(\tau)$ such that the only non-trivial solutions of the equation \eqref{Feq0} are given by the first two equations in \eqref{sexpansions}.
\DETAILS{\begin{equation*}\begin{cases}
    \lambda_s =1+ g_\lambda(|s|^2), \\
    \psi_s = s\psi_0 + sg_\psi(|s|^2).
\end{cases}\end{equation*}}
We now define the function 
    $\tilde{g}_a(s) = \al(\psi_s),$ 
where, recall, $\al(\psi)$ is defined in \eqref{eq:apsi}. This function is real-analytic and satisfies
$ \tilde{g}_a(-s) = \al(-\psi_s 
    ) = \tilde{g}_a(s),$ 
and therefore is really a function of $s^2,\ g_a(s^2)$. Define $a_s = A_0 + g_a(s^2)$. Then
 $ (\lambda_s, \psi_s, \alpha_s)$, $s \in [0, \sqrt{\e})$, solve  the rescaled Ginzburg-Landau equations \eqref{eq:rGL}.

We identify \eqref{sexpansions} with \eqref{epsexpansions} of Proposition \ref{prop:leadingorder}, with $\e=s$ and  $n=1$.  Then \eqref{lam1} implies 
\eqref{glambda'}.  
\end{proof}
Note that the definition $\lam=\frac{\kappa^2}{b}$ ($n=1$), the first equation \eqref{sexpansions} and the relation \eqref{glambda'} imply that for $(\kappa^2 - \frac{1}{2} ) \beta (\tau)+\frac{1}{2}\ge 0$, the bifurcated solution exists for $b\le \kappa^2$, and for $(\kappa^2 - \frac{1}{2} ) \beta (\tau)+\frac{1}{2}<0$, it exists for $b>\kappa^2$. Thus Theorem \ref{thm:bifurcation-resultN1}, after rescaling to the original variables,  implies (I) - (II) of Theorem \ref{thm:main-result}.

Recall that $b= \frac{\kappa^2}{\lambda}.$ 
Since the function $\lam_s=1+g_\lambda(s^2)$ given in Theorem \ref{thm:bifurcation-resultN1}, obeys  $g_\lambda(0) = 0$ and $g_\lambda '(0) \neq 0$, provided $(\kappa^2 - \frac{1}{2} ) \beta (\tau)+\frac{1}{2}\ne 0$, the function $b = \kappa^2 \lam_s^{-1}$ 
can be inverted to obtain 
$s = s(b)$. 
We can define the family $(\psi_{s(b)}, a_{s(b)}, \lam_{s(b)})$ of $\cL^\tau$-periodic solutions of the Ginzburg-Landau equations parameterized by average magnetic flux $b$.
Since $s(b)$ is real - analytic in $b$, so are $\psi_{s(b)}, a_{s(b)}, \lam_{s(b)}$. This proves (III) of Theorem \ref{thm:main-result}.

Due to \eqref{eps}, we can express the bifurcation parameter $s^2$ in terms of $b$ as
 \begin{align} \label{s}
 s^2=\frac{\kappa^2 - b}{\kappa^2 [(\kappa^2 -  \frac{1}{2})\beta(\tau) +\frac{1}{2}]\lan |\psi_0|^2\ran} +O((\kappa^2 - b)^2).
\end{align}
  Furthermore, the equations \eqref{Elam2} 
  implies that the energy of the state $(\psi_{s(b)}, a_{s(b)}, \lam_{s(b)})$,
\begin{equation} \label{Eb}
E_b(\tau) := \mathcal{E}_{\lambda(b)}(\psi_{s(b)}, a_{s(b)}), 
\end{equation} where we display the dependence on $\tau$, coming through the solution $(\psi_{s}, a_{s})$, 
has the following form
    \begin{equation}     \label{Ebeta}
        E_b(\tau) = 
        \frac{\kappa^2}{2} +b^2 -  \frac{(\kappa^2 - b)^2}{ (2\kappa^2 -  1)\beta(\tau) +1}   +  O((\kappa^2 - b)^3).
           \end{equation}

The next result addresses the nature of dependence of solution on $\tau$.
\begin{lemma} \label{lem:tauanal}
        	 $ (\lambda_s, {\psi}_s, {a}_s)$ depend smoothly on $\tau_1:=\re\tau$ and $\tau_2:=\im\tau$.
    \end{lemma}
   \begin{proof}  By the above we can write 
  $\psi_s =s\psi_0+ w(\lambda_s, s)$ and   $a_s=A_0+\al(w(\lambda_s, s), s)$, where $\al(w, s)$ is given by
  \begin{equation} \label{as}
\alpha(w, s) = (M + |s\psi_0+ w|^2)^{-1}\Im((\overline{s\psi_0+  w})\Cov{A_0}(s\psi_0+ w)),     
\end{equation}
 $w(\lambda, s)$  solves the 
   the equation
\begin{equation} \label{ws}
 ( \bar L - \lambda) w = -Q[  2i\alpha(w, s)\cdot\nabla_{A_0}(s\psi_0+ w) + |\alpha(w, s)|^2(s\psi_0+ w) + \kappa^2|(s\psi_0+ w)|^2(s\psi_0+ w)],
    \end{equation} 
and  $\lam_s$ solves the equation $\gamma(\lambda, s)=0$ (see \eqref{bif-eqn'}). Here 
$\bar L$ is the restriction of $L$ to $\Ran Q$. From their explicit expressions we see that $\psi_0$, $P$ and $Q$ are smooth in $\tau_i$. Differentiating \eqref{ws} with respect to $\tau_i$ and solving the resulting linear equation for $\p_{\tau_i} w_s$, it is not hard to convince oneself that $w_s$ is differentiable in $\tau_i$. Repeating this proceedure, one sees that $w_s$ is smooth in $\tau_i$. Therefore $\psi_s$ and $a_s$ are  smooth in $\tau_i$. For the same reason, $\gamma(\lambda, s):= P F(\lambda, s\psi_0 + w(\lambda, s))$ is smooth in $\tau_i$ and therefore so is $\lam_s$. \end{proof}
\begin{corollary} \label{lem:tauanal2}
        	 $ E_b(\tau)$  depends smoothly on $\tau_1:=\re\tau$ and $\tau_2:=\im\tau$.
    \end{corollary}
   \begin{proof} To get rid of dependence of the domain of integration in \eqref{rEnergy} on  $\tau$, we change the varables of integration in \eqref{rEnergy} as $x=m_\tau y$, where $m_\tau = (\sqrt{\Im\tau})^{-1}\left( \begin{array}{cc} 1 & \Re\tau \\ 0 & \Im\tau \end{array} \right)$, which reduces the integral to the one over the unit square cell. Since the rescaled functions,
	\begin{equation} \label{Utau}
		\begin{cases}
		\tilde{\psi}_s(x) =  \psi_s( m_\tau x), \\
		\tilde{a}_s(x) =  m_\tau^t a_s( m_\tau x), \\
		\end{cases}
	\end{equation}
defined on a $\tau$-independent square lattice, are still smooth in $\tau_1:=\re\tau$ and $\tau_2:=\im\tau$, the result follows.
\end{proof}

The next result establishes a relation between the minimizers of the energy and Abrikosov function.
\begin{theorem}\label{ABR:thm1}
    In the case $\kappa > \frac{1}{\sqrt{2}}$, the minimizers,  $\tau_b$, of $\tau \mapsto E_b(\tau)$ are related to the minimizer,  $\tau_*$, of $\beta(\tau)$, as $\tau_b - \tau_* =O(\mu^{1/2})$, where  $\mu:= \kappa^2 - b$. In particular, $\tau_b \to \tau_*$ as $b \to \kappa^2$.
\end{theorem}
\begin{proof}
To prove the theorem we note that  $E_b(\tau)$ is of the form $E_b(\tau) = e_0 + e_1 \mu + e_2(\tau) \mu^2 + O(\mu^3)$. The first two terms are constant in $\tau$, so we consider $\tilde{E}_b(\tau) = e_2(\tau) + O(\mu)$. $\tau_b$ is also the minimizer of $\tau \mapsto \tilde{E}_b(\tau)$ and $\tau_*$, of $e_2(\tau)$. We have the expansions  $\tilde{E}_b(\tau_*)-\tilde{E}_b(\tau_b) = \frac{1}{2}\tilde{E}^{''}_b(\tau_b)(\tau_*-\tau_b)^2 + O((\tau_*-\tau_b)^3)$ and  $\tilde{E}_b(\tau_*)-\tilde{E}_b(\tau_b) = -\frac{1}{2}e^{''}_2(\tau_b)(\tau_*-\tau_b)^2 + O((\tau_*-\tau_b)^3) + O(\mu)$, which imply the desired result. 
\end{proof}

The following result was discovered numerically in the physics literature and proven in \cite{ABN} using earlier result of \cite{NV}:

\begin{theorem}
\label{ABR:thm2}
    The function $\beta(\tau)$ has exactly two critical points, $\tau = e^{i\pi/3}$ and $\tau = e^{i\pi/2}$. The first is minimum, whereas the
    second is a maximum.
\end{theorem}

Theorems 
\ref{ABR:thm1} and \ref{ABR:thm2} imply the remaining statement, (IV), of 
Theorem \ref{thm:main-result}. $\qed$

\noindent\textbf{Remarks.} 1) The applied magnetic field is given by $h_0=\frac{1}{2}\p_bE_b(\tau)$. 
From  \eqref{Ebeta} we have formally
\begin{align}     \label{h0}
h_0=b +  \frac{\kappa^2 - b}{ (2\kappa^2 -  1)\beta(\tau) +1}   +  O((\kappa^2 - b)^2)\\
        =b +  \frac{\kappa^2 }{ 2}\lan |\psi_0|^2\ran   +  O((\kappa^2 - b)^2).
           \end{align}
\DETAILS{2) Absorbing $s$ into $\psi_0$, we derive from the first equation in \eqref{sexpansions}, \eqref{glambda'} and \eqref{s} (a correct form of) the classical Abrikosov's relation
   \begin{align} \label{Arel}
		\frac{\kappa^2-b}{\kappa^2}\lan |\psi_0|^2\ran - \frac{1}{2} \lan |\psi_0|^2\ran^2= \left(\kappa^2 - \frac{1}{2} \right) \lan |\psi_0|^4\ran.
	\end{align}}

2) The proof of Theorem \ref{thm:bifurcation-resultN1} gives in fact the following abstract result.
\begin{theorem}
\label{thm:bifurcation-theorem-1}
    Let $X$ and $Y$ be complex Hilbert spaces, with $X$ a dense subset of $Y$, and consider a map $F : \R \times X \to Y$ that
    is analytic as a map between real Banach spaces. Suppose that for some $\lambda_0 \in \R$, the following conditions are satisfied:
    \begin{enumerate}[(1)]
        \item \label{BT:trivial} $F(\lambda, 0) = 0$ for all $\lambda \in \R$,
        \item \label{BT:L-0} $\p_\psi F(\lambda_0, 0)$ is self-adjoint and has an isolated eigenvalue at $0$ of (geometric) multiplicity $1$,
		\item \label{BT:v-star} For non-zero $\psi_0 \in \Null \p_\psi F(\lambda_0, 0)$,
                    $\langle \psi_0, \p_{\lambda,\psi}F(\lambda_0, 0)\psi_0 \rangle \neq 0$,
        \item \label{BT:sym} For all $\alpha \in \R$, $F(\lambda, e^{i\alpha}\psi) = e^{i\alpha}F(\lambda, \psi)$.
		\item \label{BT:SA} For all $\psi \in X$, $\langle \psi, F(\lambda, \psi) \rangle \in \R$.
    \end{enumerate}
    Then $(\lambda_0, 0)$ is a bifurcation point of the equation $F(\lambda, \psi) = 0$, in the sense that there is a family of non-trivial solutions, $(\lambda_s, \psi_s)$,  for $s \in [0, \sqrt{\e})$,   
    unique modulo the global gauge symmetry  (apart from the trivial solution $(1,0)$) in a neighbourhood of $(\lambda_0, 0)$ in $\R \times X$. Moreover, this family has the form
    \begin{equation*}
    \begin{cases}
        \lambda = \phi_\lambda(s^2), \\
        \psi = s\psi_0 + s\phi_\psi(s^2).
    \end{cases}
    \end{equation*}
    Here $\psi_0 \in \Null \p_\psi F(\lambda_0, 0)$, 
     and  $\phi_\lambda : [0,\epsilon) \to \R$ and $\phi_\psi : [0,\epsilon) \to X$ are unique real-analytic functions,
    such that $\phi_\lambda(0) = \lambda_0$, $\phi_\psi(0) = 0$. 
\end{theorem}

\appendix


\section{Fixing the Gauge}
\label{sec:alternate-proof}

We provide here an alternate proof of Proposition \ref{thm:fix-gauge}, largely based on ideas in \cite{Eil}. We begin by defining the function $B : \R \to \R$ to be
	\begin{align*}
		B(\zeta) = \frac{1}{r} \int_{0}^{r} \Curl A(\xi, \zeta)\,d\xi.
	\end{align*}
	It is clear that $b = \frac{1}{r\tau_2} \int_{0}^{r\tau_2} B(\zeta) \,d\zeta$. A calculation shows that $B(\zeta + r\tau_2) = B(\zeta)$.
	

	We now define $P = (P_1, P_2) : \R^2 \to \R^2$ to be
	\begin{align*}
	& P_1(x) = bx_2 - \int_{0}^{x_2} B(\zeta) \,d\zeta, \\
	& P_2(x) = \int_{\frac{\tau_1}{\tau_2} x_2}^{x_1} \Curl A(\xi, x_2)\,d\xi + \frac{\tau \wedge x}{\tau_2} B(x_2).
	\end{align*}
	A calculation shows that $P$ is doubly-periodic with respect to $\Lat$.

		
	We now define $\eta' : \R^2 \to \R$ to be
	\begin{align*}
		\eta'(x) = \frac{b}{2}x_1x_2 - \int_{0}^{x_1} A_1(\xi, 0) \,d\xi - \int_{0}^{x_2} A_2(x_1, \zeta) - P_2(x_1, \zeta) \,d\zeta.
	\end{align*}
	$\eta'$ satisfies 
		$\nabla\eta' = - A + A_0 + P.$  

	Now let $\eta''$ be a doubly-periodic solution of the equation $\Delta\eta'' = -\Div P$. Also let $C = (C_1, C_2)$ be given by
	\begin{equation*}
		C = - \frac{1}{|\Omega|} \int_\Omega (P + \nabla\eta'') \,dx,
	\end{equation*}
	where $\Omega$ is any fundamental cell, and set $\eta''' = C_1x_1 + C_2x_2$.
	
	We claim that $\eta = \eta' + \eta'' + \eta'''$ is such that $A + \nabla\eta$ satisfies (i) - (iii) of the proposition. We first note that $A + \nabla\eta = A - A + A_0 + P + \nabla\eta'' + C$. By the above, $A' = P + \nabla\eta'' + C$ is periodic. We also calculate that $\Div A' = \Div P + \Delta\eta'' = 0$. Finally $\int_\Omega A' = \int_\Omega (P + \nabla\eta - C) = 0$.
	
	All that remains is to prove (iv). This will follow from a gauge transformation and translation of the state. We note that
	\begin{equation*}
		A_0(x + t) + A'(x + t) = A_0(x) + A'(x) + \frac{b}{2} \left( \begin{array}{c} -t_2 \\ t_1 \end{array} \right).
	\end{equation*}
	This means that $A_0(x + t) + A'(x + t) = A_0(x) + A'(x) + \nabla g_t(x)$, where $g_t(x) = \frac{b}{2} t \wedge x + C_t$ for some constant $C_t$. To establish (iv), we need to have it so that $C_t = 0$ for $t = r$, $r\tau$. First let $l$ be such that $r \wedge l = -\frac{C_r}{b}$ and $r\tau \wedge l = -\frac{C_{r\tau}}{b}$. This $l$ exists as it is the solution to the matrix equation
	\begin{equation*}
		\left( \begin{array}{cc} 0 & r \\ -r\tau_2 & r\tau_1 \end{array} \right)
			\left( \begin{array}{c} l_1 \\ l_2 \end{array} \right) =
			\left( \begin{array}{c} -\frac{C_r}{b} \\ -\frac{C_{r\tau}}{b} \end{array} \right),
	\end{equation*}
	and the determinant of the matrix is just $r^2\tau_2$, which is non-zero because $(r,0)$ and $r\tau$ form a basis of the lattice. Let $\zeta(x) = \frac{b}{2} l \wedge x$. A straight forward calculation then shows that $e^{i\zeta(x)}\psi(x + l)$ satisfies (iv) and that $A(x + l) + \nabla\zeta(x)$ still satisfies (i) through (iii). This proves the proposition.



\end{document}